\documentclass{article}

\usepackage{microtype}
\usepackage{graphicx}
\usepackage{subfigure}
\usepackage{booktabs} 
\usepackage{subcaption} 
\usepackage{multirow}
\usepackage{makecell}

\newcommand{\zhb}[1]{{\color{black}#1}}

\usepackage{hyperref}
\usepackage{algorithm}
\usepackage{algpseudocode}

\usepackage[preprint]{paper}

\usepackage{amsmath}
\usepackage{amssymb}
\usepackage{mathtools}
\usepackage{amsthm}
\usepackage{balance}

\usepackage[capitalize,noabbrev]{cleveref}

\theoremstyle{plain}
\newtheorem{theorem}{Theorem}[section]

\theoremstyle{definition}

\theoremstyle{remark}

\usepackage[textsize=tiny]{todonotes}

\icmltitlerunning{PipeMax: Enhancing Offline LLM Inference on Commodity GPU Servers}

\begin{document}

\twocolumn[
\icmltitle{PipeMax: Enhancing Offline LLM Inference on Commodity GPU Servers}



\icmlsetsymbol{equal}{*}

\begin{icmlauthorlist}

\icmlauthor{Hongbin Zhang}{sysu}
\icmlauthor{Taosheng Wei}{sysu}
\icmlauthor{Jiazhi Jiang}{sysu}
\icmlauthor{Hui Yan}{sysu}
\icmlauthor{Jiangsu Du}{sysu}
\icmlauthor{Zhiguang Chen}{sysu}

\end{icmlauthorlist}

\icmlaffiliation{sysu}{School of Computer Science and Engineering, Sun Yat-sen University}



\vskip 0.3in
]



\printAffiliationsAndNotice{}  

\begin{abstract}
Offline LLM inference seeks to maximize request processing under fixed budgets, making commodity GPU servers a promising choice.
However, prior work typically considers offloading and parallelism in isolation, resulting in suboptimal performance.
In this paper, we propose PipeMax, a high-throughput LLM inference system that integrates pipeline parallelism with offloading to overcome interconnect and memory constraints on GPU servers.
Particularly, pipeline parallelism naturally incurs low communication overhead and keeps only one batch active on each GPU at a time, which enables offloading the KV cache of inactive batches.
By coordinating computation with offloading data movement, PipeMax effectively expands GPU memory capacity and sustains large-batch execution.
Experiments show that PipeMax achieves up to 2.51× higher throughput than vLLM, and up to 1.42× and 1.38× higher throughput than state-of-the-art high-throughput LLM systems, respectively, on an 8-GPU node.

\end{abstract}

\section{Introduction}
Large language models (LLMs) have been widely adopted across many domains~\cite{copilot,nazi2024largelanguagemodelshealthcare,rane2023contribution}. 
Due to their massive parameter scales, LLM inference requires numerous high-cost GPUs, making it extremely expensive to deploy.
Beyond interactive applications such as chatbots, LLMs are increasingly used in offline workloads, including database processing~\cite{liu2025optimizing} and information extraction~\cite{xu2024large}. Unlike online LLM serving~\cite{zhong2024distserve, agrawal2024taming, du2025ecoserve}, which prioritizes latency SLOs, offline scenarios primarily target high-throughput execution.

Providing high-throughput LLM inference on commodity GPU servers without high-bandwidth interconnects is increasingly important.
First, commodity GPU servers such as RTX~5090 are approximately 3× more cost-effective under the same computational capability~\cite{feng2023mobius} and constitute a large fraction of the deployed GPU infrastructure~\cite{du2025ecoserve}.
Second, high-end GPUs with high-bandwidth interconnects (e.g., NVLink), such as H100 and B200, are in practice preferentially reserved for latency-critical LLM serving.
However, efficiently exploiting commodity GPU servers for high-throughput LLM inference remains challenging.

LLM inference imposes heavy GPU memory demands due to large model weights and extensive KV cache, which often exceeds the footprint of the weights.
Existing systems~\cite{sheng2023flexgen, zhang2025td} attempt to expand memory capacity via either offloading or parallelization.
Purely offloading-based approaches are fundamentally constrained by limited CPU–GPU bandwidth that transferring model weights and KV cache dominates execution time, leaving GPUs largely underutilized.
Parallelism-based approaches include tensor parallelism and pipeline parallelism.
Among them, tensor parallelism is communication-bound due to frequent all-reduce operations, rendering it impractical on bandwidth-limited nodes.
Consequently, pipeline parallelism emerges as the most promising option for commodity GPU servers.

However, standalone pipeline parallelism falls short of realizing its full memory efficiency.
During the decode stage under pipeline parallelism, only the KV cache of a single batch is active for each GPU at any given time, leaving most GPU memory occupied by inactive batches and limiting effective memory expansion.
To address this limitation, we propose PipeMax, a high-throughput LLM inference system tailored for commodity GPU servers.
PipeMax enhances pipeline parallelism with offloading by storing inactive batches in CPU memory, allowing pipeline parallelism to fully exploit its memory efficiency.
In summary, we make the following contributions:

\begin{itemize}
    \vspace{-10pt}
    \item We identify a decode-phase inefficiency in pipeline parallelism: KV cache remains idle across inactive batches, limiting effective GPU memory utilization.
    \vspace{-5pt}
    \item We propose PipeMax, a high-throughput LLM inference system that strategically integrates pipeline parallelism with offloading to expand effective GPU memory by evicting inactive KV cache.
    \vspace{-5pt}
    \item We evaluate PipeMax and demonstrate it has substantial throughput improvements over state-of-the-art systems across a range of model sizes and workloads.
\end{itemize}

\vspace{-15pt}
\section{Background}
\subsection{LLM Processing Phases}
\label{LLM}
\vspace{-2pt}

Recent large language models (LLMs) adopt an autoregressive generation paradigm.
As shown in Fig.~\ref{fig:autoregressive}, the model iteratively predicts tokens until the end-of-sequence (EoS), storing intermediate states as a KV cache to avoid redundant computation.
Accordingly, LLM inference consists of two phases: prefill for processing the input sequence, and decode for sequential token generation using the KV cache.

Prior work~\cite{du2025ecoserve,jiang2025efficient} shows that prefill is compute-intensive and efficient even with small batches, whereas decode is memory-bound and requires large batches for peak utilization.
However, achieving large decode batches requires substantial GPU memory for KV cache. For instance, a batch size of 512 with a sequence length of 1024 in Qwen3-32B requires about 133 GB of KV cache, far beyond a single GPU’s capacity.
Thus, decode throughput is fundamentally limited by GPU memory.

\vspace{-2pt}
\subsection{Existing High-Throughput LLM Inference}
\vspace{-2pt}

To address above GPU memory issue, existing systems~\cite{zhang2025td, kwon2023efficient, su2025seesaw, sheng2023flexgen} commonly use offloading or model parallelism to expand GPU memory.

\vspace{-2pt}
\subsubsection{Offloading Approach}
\vspace{-2pt}
Systems such as FlexGen~\cite{sheng2023flexgen} and TightLLM~\cite{hu2025tightllm} enable single-GPU high-throughput LLM inference by offloading model weights and/or KV cache to CPU memory.
As illustrated in Fig.~\ref{fig:offload-tech}, they rely on two core techniques: layer-wise offloading, which keeps only a  subset of layers on the GPU, and multi-batch inference, which partitions requests into multiple batches to reduce the KV cache footprint.
Both techniques use prefetching to overlap computation and data transfer.

However, offloading-only approaches remain bandwidth-bound on commodity GPUs: limited PCIe bandwidth causes weight transfers to dominate execution, and on multi-GPU nodes, repeatedly transferring weights to each GPU leads to redundant communication and wasted bandwidth.

\vspace{-2pt}
\subsubsection{Model Parallelism Approach}
\vspace{-2pt}

The two mainstream model parallelism techniques are tensor parallelism and pipeline parallelism, both of which distribute model parameters across multiple GPUs to expand total memory capacity.
As shown in Fig.~\ref{fig:model-parallelism}, tensor parallelism partitions computation within each layer and requires all-reduce synchronization at every layer, whereas pipeline parallelism assigns different layer groups to different GPUs and transfers intermediate activations between stages.

\begin{figure}[!t]
    \centering    
    \includegraphics[width=0.7\linewidth]{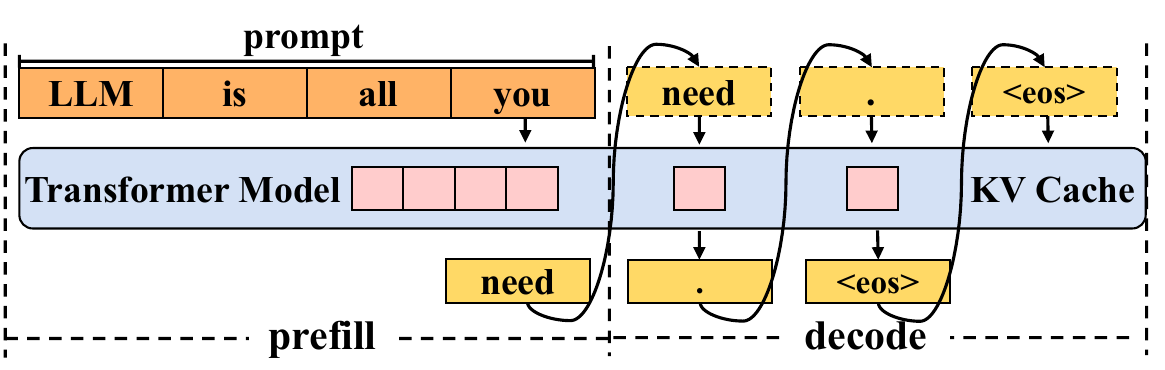}
    \vspace{-3pt}
    \caption{Autoregressive generation in LLM inference.}
    \label{fig:autoregressive}
    \vspace{-10pt}
\end{figure}

\begin{figure}[!t]
    \centering    \includegraphics[width=0.9\linewidth]{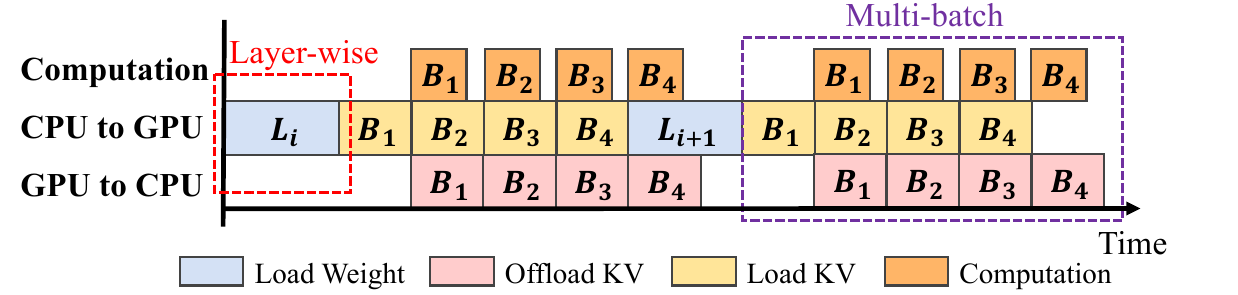}
    \vspace{-2pt}
    \caption{Existing offload-based LLM inference method.}
    \label{fig:offload-tech}
    \vspace{-15pt}
\end{figure}
\begin{figure}[!b]
    \vspace{-14pt}
    \centering    
    \includegraphics[width=0.88\linewidth]{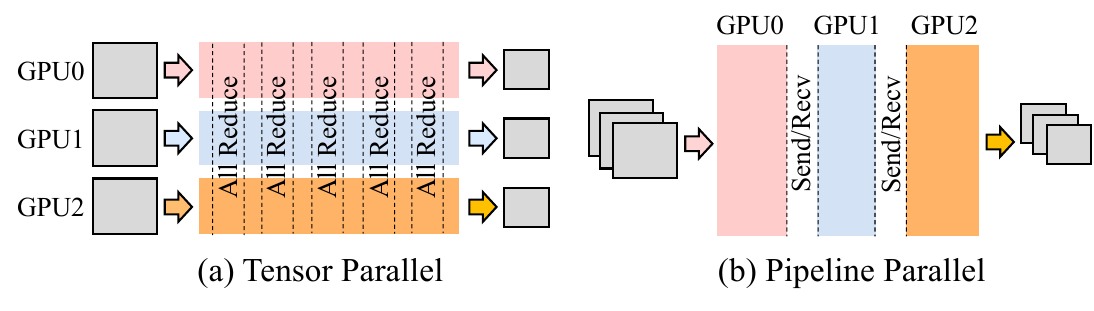}
    \vspace{-6pt}
    \caption{Tensor parallelism and pipeline parallelism.}
    \vspace{-2pt}
    \label{fig:model-parallelism}
\end{figure}

Prior work~\cite{zhang2025td,su2025seesaw} shows that tensor parallelism incurs high communication overhead on GPU servers without high-speed interconnects due to frequent all-reduce operations, thereby making pipeline parallelism a natural alternative.
Motivated by this, TD-Pipe~\cite{zhang2025td} adopts temporally disaggregated pipeline parallelism with inter-batch work stealing to mitigate pipeline bubbles and improve decode-phase arithmetic intensity.
Meanwhile, Seesaw~\cite{su2025seesaw} re-shards the model across prefill and decode, using pipeline parallelism for prefill and tensor parallelism for decode.
However, both approaches have limitations.
TD-Pipe retains multiple decode batches in GPU memory, leaving much KV cache inactive (Section~\ref{decode-phase}),
while Seesaw assumes negligible decode-phase communication overhead, which is often non-negligible on bandwidth-constrained commodity servers, and becomes increasingly restrictive as the system scales to more GPUs.

\vspace{-2pt}
\subsection{Analysis of Pipeline Parallel for LLM Inference}
\vspace{-2pt}
\label{pp-analyse}
\subsubsection{Prefill phase}
\label{pp-analyse-prefill}
\vspace{-2pt}
Fig.~\ref{fig:pure-prefill} shows that under continuous pipeline-parallel execution, prefill latency is dominated by the first pipeline stage for large request counts.
The total execution time equals the first-stage time plus $(n-1)$ times the per-stage prefill time of the longest request, where $n$ is the number of GPUs (Appendix~\ref{app:prefill-proof}).
As prefill is compute-intensive and free of inter-request dependencies, pipeline parallelism achieves high GPU utilization with minimal overhead.
\vspace{-2pt}
\subsubsection{Decode phase}
\vspace{-2pt}
\label{decode-phase}

In contrast, decode exhibits fundamentally different behavior under pipeline parallelism.
Prior work~\cite{zhang2025td, zhong2024distserve} partitions GPU-resident requests into multiple autoregressive batches to fill the pipeline, but inter-step dependencies cause execution-time variations to amplify into inter-batch imbalance, leading to pipeline stalls (Fig.~\ref{fig:decode-pp}(a)).
TD-Pipe~\cite{zhang2025td} mitigates this issue via inter-batch work stealing (Fig.~\ref{fig:decode-pp}(b)).
\begin{figure}[!t]
    \vspace{-2pt}
    \centering    \includegraphics[width=0.9\linewidth]{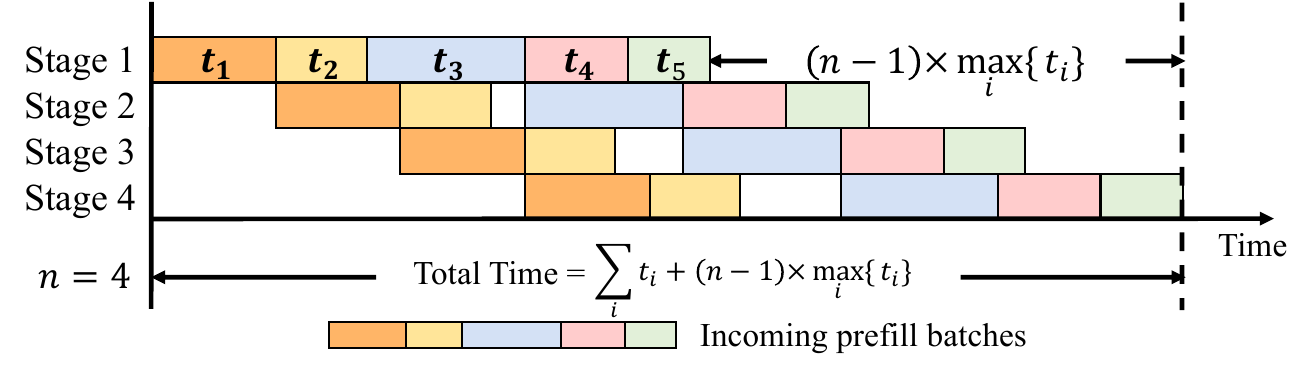}
    \vspace{-2pt}
    \caption{Prefill phase under pipeline parallelism. The total execution time equals the first stage plus $(n - 1)$ times the longest stage, where $n$ is the number of GPUs.}
    \vspace{-12pt}
    \label{fig:pure-prefill}
\end{figure}

Despite such balancing, decode throughput under pipeline parallelism remains constrained by the memory footprint of the \emph{active} batch.
At any time, each GPU executes only one active decode batch, leaving the KV cache of other batches idle and resulting in low effective GPU memory utilization.
The average token budget per batch can be expressed as
\vspace{-3pt}
\begin{equation}
\frac{M - W/n}{T},
\label{equation: memory}
\end{equation}

\vspace{-5pt}
where $M$ denotes the per-GPU memory capacity, $W$ the total model size, $n$ the pipeline degree, and $T$ the KV cache size per token; a detailed derivation is provided in Appendix~\ref{app:decode-proof}.
As a result, although pipeline parallelism increases the \emph{aggregate} KV cache capacity across devices, the KV cache available to the \emph{active} batch remains fundamentally constrained.
Importantly, multi-batch execution exhibits natural \emph{anti-locality}: a batch’s KV cache remains unused until all other batches finish their decode iterations.

\vspace{-2pt}
\subsubsection{Prefill-Decode imbalance}
\vspace{-2pt}

\begin{figure}[!b]
    \vspace{-10pt}
    \centering    \includegraphics[width=0.99\linewidth]{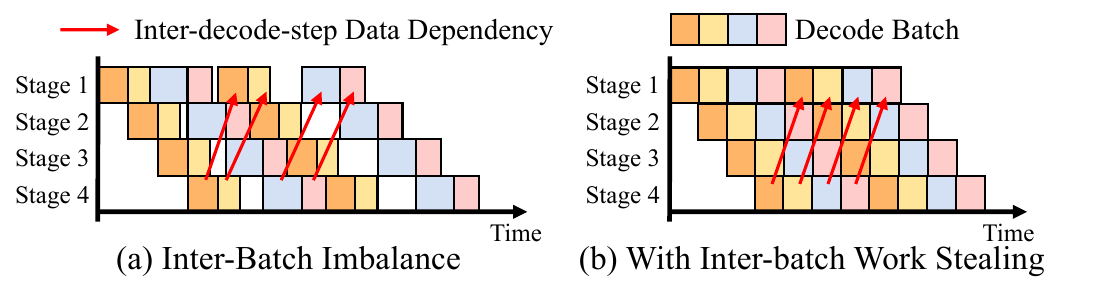}
    \vspace{-8pt}
    \caption{Decode phase under pipeline parallelism. }
    \label{fig:decode-pp}
    \vspace{-10pt}
\end{figure}
\begin{figure}[!t]
    \vspace{-2pt}
    \centering    \includegraphics[width=0.85\linewidth]{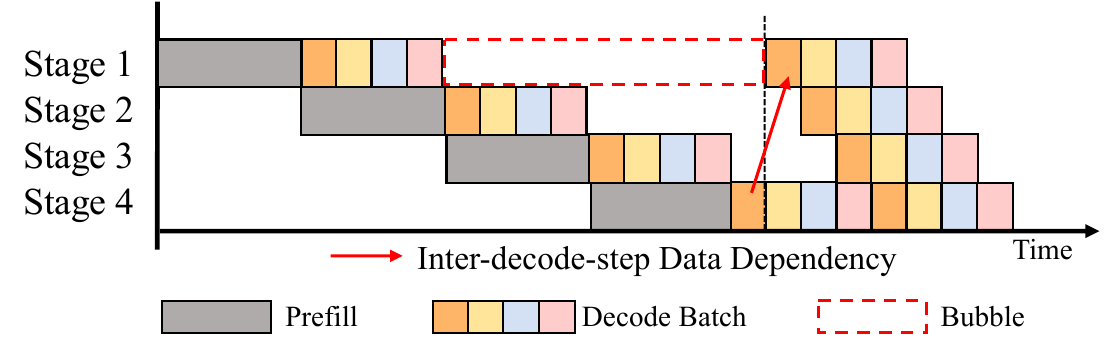}
    \vspace{-2pt}
    \caption{Prefill-decode imbalance. }
    \vspace{-15pt}
    \label{fig:prefill-decode}
\end{figure}
Frequent transitions between prefill and decode introduce pipeline bubbles due to execution-time mismatch, as illustrated in Fig.~\ref{fig:prefill-decode}.
TD-Pipe~\cite{zhang2025td} mitigates this overhead through temporal separation, which PipeMax adopts as its baseline execution model.
\vspace{-2pt}
\subsection{Opportunities and Challenges}
\label{opportunities}
\vspace{-2pt}
The anti-locality in Section~\ref{decode-phase}, together with pipeline parallelism’s multi-batch structure, naturally motivates offloading inactive KV cache to CPU memory and prefetching it on demand, thereby expanding the effective GPU memory for the active batch.
Fig.~\ref{fig:decode-prefetch} illustrates this mechanism.
However, realizing this opportunity introduces three challenges:

\vspace{-12pt}
\paragraph{C1: Uncertain Compute–Prefetch Overlap.} 
Achieving effective KV cache prefetching requires overlapping prefetch with decode computation, so that GPU memory only needs to hold the KV cache of the current and upcoming batches.
However, this ideal overlap is not always attainable in practice.
Limited CPU–GPU bandwidth may prevent fully prefetching a batch’s KV cache within a single decode iteration, delaying subsequent iterations and causing GPU idle time.
Moreover, decode execution time varies across iterations due to irregular request completion and dynamic batch composition, further complicating compute–communication overlap.

\vspace{-12pt}
\paragraph{C2: Prefetch-Aware Decode Batch Balancing.}
As discussed in Section~\ref{decode-phase}, efficient decode execution requires balanced execution times across batches.
TD-Pipe achieves this via inter-batch work stealing under the assumption of a closed set of GPU-resident requests.
With prefetching, however, decode batches are formed from a broader pool of CPU-resident requests, invalidating this assumption.
Consequently, batch balancing must jointly consider execution time and prefetch feasibility, making it significantly harder to maintain stable execution across iterations.

\vspace{-12pt}
\paragraph{C3: KV Cache Transfer Inefficiency under PagedAttention.}
Modern LLM inference frameworks adopt PagedAttention~\cite{kwon2023efficient} to improve GPU memory utilization, yet mainstream implementations are not transfer-friendly.
In systems such as vLLM~\cite{kwon2023efficient} and sglang~\cite{zheng2024sglang}, KV cache is organized in page-sized blocks and further separated along the layer dimension, fragmenting each request’s KV cache along two dimensions and leading to inefficient CPU–GPU transfers.

\vspace{-8pt}
\section{The Design of PipeMax}
\vspace{-2pt}
In this section, we propose PipeMax, a high-throughput LLM inference system that leverages pipeline parallelism with offloading.
Specifically, PipeMax leverages pipeline parallelism to partition model weights across GPUs, while offloading only KV cache to CPU memory.

\vspace{-5pt}
\subsection{PipeMax Workflow}
\vspace{-2pt}

As illustrated in Fig.~\ref{fig:pipemax-workflow}, PipeMax temporally decouples prefill and decode, establishing a producer--consumer pipeline between the two stages.
\begin{figure}[!t]
    \centering    \includegraphics[width=0.8\linewidth]{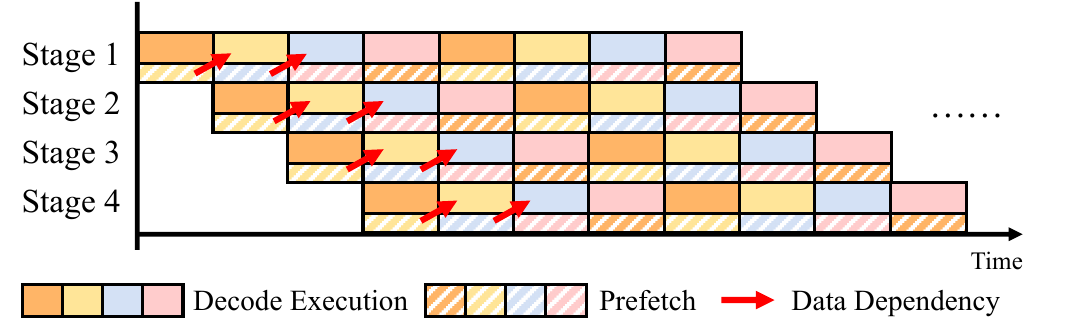}
    \vspace{-5pt}
    \caption{Decode phase with KV cache prefetching.}
    \label{fig:decode-prefetch}
    \vspace{-17pt}
\end{figure}
\begin{figure}[!b]
    \vspace{-20pt}
    \centering
    \includegraphics[width=0.87\linewidth]{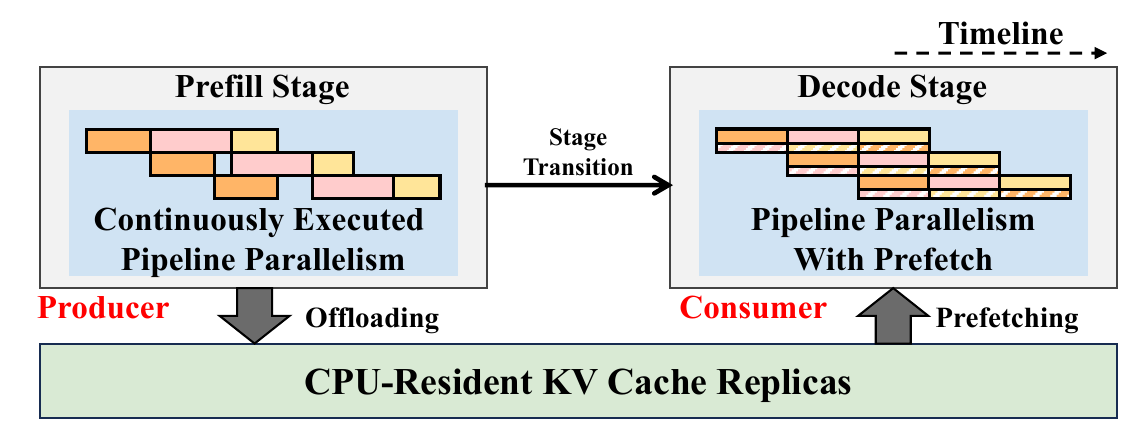}
    \vspace{-4pt}
    \caption{The workflow of PipeMax.}
    \label{fig:pipemax-workflow}
\end{figure}

\vspace{-3pt}
\subsubsection{Prefill Stage}
\vspace{-2pt}
As discussed in Section~\ref{pp-analyse}, PipeMax continuously performs prefill to maximize pipeline utilization.
The generated KV cache is asynchronously offloaded to CPU memory, allowing GPU memory to be reused by overwriting completed requests without blocking.
This decouples prefill from GPU memory constraints and accumulates decode requests.

\vspace{-2pt}
\subsubsection{Decode Stage}
\label{workflow-decode}
\vspace{-1pt}

Once sufficient requests are buffered, PipeMax enters the decode stage and addresses the challenges identified above.

\vspace{-1pt}
\textbf{C1.} PipeMax adopts a best-effort KV cache prefetching strategy.
Before each iteration, it estimates decode execution time to determine the prefetch budget and incrementally prefetches KV cache by overwriting inactive batches.
When bandwidth is insufficient, prefetching is amortized across iterations, resulting in a hybrid GPU-resident and prefetched KV cache per batch (Fig.~\ref{fig:kvcache-composition}).

\vspace{-1pt}
\textbf{C2.} PipeMax employs a prefetch-aware scheduler that jointly considers execution time and prefetch feasibility, aligning batch execution times to mitigate inter-batch imbalance.

\vspace{-1pt}
\textbf{C3.} PipeMax further introduces a transfer-efficient KV cache engine that maximizes CPU--GPU bandwidth utilization while remaining compatible with PagedAttention.
\vspace{-3pt}
\subsubsection{Prefill--Decode Switching Policy}
\vspace{-1pt}
Unlike TD-Pipe, which frequently switches between prefill and decode to sustain compute intensity under GPU memory constraints, PipeMax buffers a large pool of decode-ready requests in CPU memory and adopts a memory-driven switching policy.

Prefill proceeds until CPU-resident KV cache approaches capacity, reserving headroom for decode, and resumes only when available KV cache falls below GPU capacity. As a result, prefill--decode switching is infrequent, and the imbalance in Fig.~\ref{fig:prefill-decode} can be safely ignored.

\vspace{-3pt}
\subsection{System Overview}
\vspace{-2pt}

As shown in Fig.~\ref{fig:pipemax-overview}, PipeMax adopts a hierarchical controller architecture that separates control from execution.
A centralized engine serves as the control plane, while a distributed runtime constitutes the execution plane, jointly supporting the execution workflow described above.
We next describe the key mechanisms of both components.

\begin{figure}[!t]
    \centering    \includegraphics[width=0.75\linewidth]{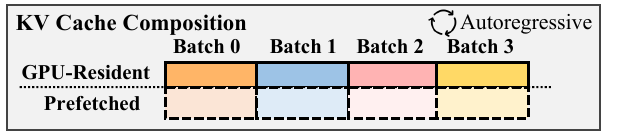}
    \vspace{-2pt}
    \caption{KV cache composition per batch in PipeMax under bandwidth constraints.}
    \vspace{-16pt}
    \label{fig:kvcache-composition}
\end{figure}
\begin{figure}[!b] 
\vspace{-14pt} 
\centering 
\includegraphics[width=0.9\linewidth]{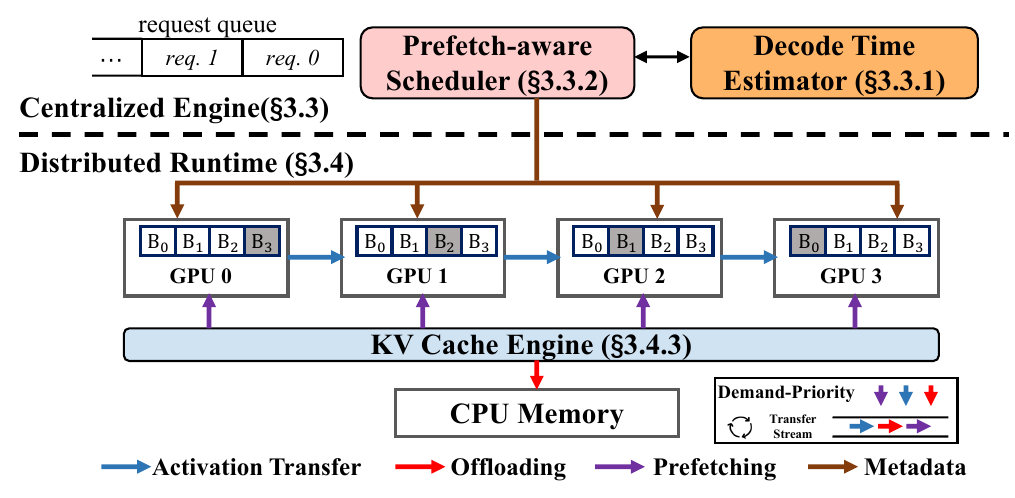} 
\vspace{-6pt} 
\caption{Overview of PipeMax system.}
\label{fig:pipemax-overview} 
\end{figure}

\vspace{-4pt}
\subsection{Centralized Engine}
\vspace{-2pt}
The centralized engine integrates a decode execution-time estimator
and a prefetch-aware scheduler to coordinate decode execution and KV cache prefetching.

\vspace{-2pt}
\subsubsection{Decode Execution Time Estimator}
\label{Estimator}
\vspace{-2pt}

To overlap decode computation with KV cache prefetching, PipeMax estimates decode execution time for each batch.

A decode iteration consists of two components: (1) linear operations (e.g., QKV projection and feed-forward networks) with cost $O(b \times h^2)$ for batch size $b$ and hidden size $h$; and (2) attention over the prefix KV cache with cost $O\!\left(L\right)$, where $L=\sum_{i=1}^{b} L_i$ denotes the total prefix length of the batch, and $L_i$ is the prefix length of request $i$.

Accordingly, PipeMax models the decode execution time of a batch as:
\begin{equation}
\label{model}
\alpha \cdot b + \beta \cdot L + \delta,
\end{equation}
where $\alpha$, $\beta$, and $\delta$ are parameters obtained via offline profiling,
capturing the per-request linear cost, per-token attention cost, and constant overheads, respectively.

\vspace{-2pt}
\subsubsection{Prefetch-aware decode scheduler}
\label{scheduler}
\vspace{-2pt}

To address C2 in Section~\ref{opportunities}, PipeMax introduces a prefetch-aware scheduler for the decode stage that dynamically determines how many requests to prefetch and which ones to select.
The scheduler aims to fully overlap computation with KV cache prefetching, while expanding effective memory capacity and maintaining balance across batches.

As discussed in Section~\ref{workflow-decode}, limited CPU--GPU bandwidth may leave portions of KV cache from multiple batches resident in GPU memory.
To efficiently utilize GPU memory while preserving autoregressive semantics, PipeMax partitions decode requests into $n$ batches, where $n$ is the pipeline depth, and formulates decode scheduling as a stateful, prefetch-aware iterative batch update problem.

PipeMax maintains a set of decode batches
\vspace{-3pt}
\begin{equation}
\mathcal{D} = \{ D_0, D_1, \ldots, D_{n-1} \},
\vspace{-5pt}
\end{equation}
which are executed autoregressively in a cyclic order.
At iteration $t$, the executing, prefetched, and overwritten batches are indexed as
\vspace{-6pt}
\begin{equation}
\begin{aligned}
i_t &= t \bmod n, \\
j_t &= (i_t + 1) \bmod n, \\
k_t &= (i_t - 1) \bmod n . 
\end{aligned}
\end{equation}

\vspace{-8pt}
During iteration $t$, PipeMax executes batch $D_{i_t}$ while concurrently prefetching KV cache from CPU
memory to update $D_{j_t}$ for execution in iteration $t+1$, overlapping computation with data movement.
\vspace{-5pt}
\paragraph{Initial Decode Batches.}
Let $\mathcal{R}$ denote the set of requests with KV cache resident in GPU memory after prefill.
PipeMax constructs an initial partition $\mathcal{D}$
such that
\vspace{-4pt}
\begin{equation}
\bigcup_{k=0}^{n-1} D_k = \mathcal{R}, \quad D_k \cap D_{k'} = \emptyset \;\; \forall k \neq k'.
\vspace{-3pt}
\end{equation}
Each subset $D_k$ forms an initial decode batch.
The initial partition divides requests into batches of equal size, while attempting to balance total KV cache length across batches to approximate similar decode execution times, thereby reducing inter-batch imbalance across pipeline stages.
\vspace{-5pt}
\paragraph{Iterative Prefetch-Aware Scheduling.}

After initialization, decode execution proceeds iteratively in an autoregressive manner.
At each iteration $t$, PipeMax updates the next decode batch $D_{j_t}$ using a prefetch-aware scheduling policy.

First, PipeMax retains requests whose KV cache remains resident in GPU memory:

\vspace{-8pt}
\begin{equation}
D_{j_t}^{\mathrm{res}} = D_{j_t} \cap \mathcal{R}_t ,
\vspace{-1pt}
\end{equation}
where $\mathcal{R}_t$ denotes the set of requests whose KV cache resides in GPU memory at the beginning of iteration $t$.

Second, PipeMax determines additional requests to prefetch from CPU memory.
It predicts the execution time $\hat{T_t}$ of the currently executing batch $D{i_t}$ using the estimator in Section~\ref{Estimator}, and derives a prefetch budget
\begin{equation}
\vspace{-3pt}
\label{eq:budget}
\mathcal{B}_t = B \cdot \hat{T}_t,
\vspace{-1pt}
\end{equation}
where $B$ denotes the effective CPU–GPU bandwidth, profiled using KV cache transfers in the PagedAttention format.

Given $\mathcal{B}_t$, PipeMax selects CPU-resident requests $P_t$
whose total prefix length approaches $\mathcal{B}_t$,
maximizing bandwidth utilization and effective GPU memory expansion.

While the above formulation specifies $P_t$ per iteration,
the scheduling policy evolves across iterations,
and the decode execution is divided into warm-up and steady phases.\\
\emph{\textbf{In warm-up phase}}, PipeMax enlarges the prefetch budget across iterations by extending decode iterations, thereby expanding effective GPU memory.
According to Eq.~\ref{model}, when the total KV cache length of the current batch $L$ is close to $\mathrm{len}(D_{j_t}^{\mathrm{res}}) + \mathcal{B}_t$, execution time is dominated by the batch size~$b$.
PipeMax therefore prioritizes short requests to pack more requests into each batch,
maximizing the batch size~$b$ under a fixed KV cache budget.
This increases the prefetch budget in subsequent iterations via the time-based update, forming a positive feedback loop.
However, as longer requests are admitted, KV cache growth under fixed GPU memory and CPU--GPU bandwidth imposes hard limits, causing execution time to fluctuate.
In practice, after such temporary fluctuations, execution time converges to a bounded range, after which the system enters the steady phase.\\
\emph{\textbf{In steady phase}}, the primary objective shifts to maintaining execution-time balance across decode batches to mitigate inter-batch imbalance.
To this end, PipeMax selects $P_t$ such that the predicted execution time of the updated batch remains close to $\hat{T}_t$.
Based on the execution-time model in Section~\ref{Estimator}, the retained set $D_{j_t}^{\mathrm{res}}$ contributes a deterministic execution time
$\hat{T}^{\mathrm{res}}_t$, computed using Eq.~(\ref{model}).
The remaining execution-time gap is

\vspace{-8pt}
\begin{equation}
\label{cpu-time}
\Delta \hat{T}_t = \hat{T}_t - \hat{T}^{\mathrm{res}}_t .
\end{equation}
\vspace{-17pt}

Selecting $P_t$ is formulated as a subset selection problem, over CPU-resident requests, where each request $r$ contributes $\alpha + \beta L_r$ according to Eq.~(\ref{model}). The cumulative prefix length is constrained to nearly saturate the prefetch budget $\mathcal{B}_t$, so as to fully utilize the available CPU--GPU bandwidth and maximally extend the effective GPU memory capacity. 
The objective is to make the execution-time contribution of CPU-resident requests as close as possible to $\Delta \hat{T}_t$. 
PipeMax provides a greedy algorithm to efficiently solve this problem; the complete scheduling procedure and the detailed algorithm for selecting $P_t$ are presented in Appendix~\ref{app:scheduler}.

During prefetching for $D_{j_t}$, PipeMax reuses GPU memory by overwriting KV cache blocks associated with the inactive batch $D_{k_t}$.
The updated decode batch is
\begin{equation}
D_{j_t} = D_{j_t}^{\mathrm{res}} \cup P_t.
\end{equation}
\vspace{-25pt}
\subsection{PipeMax Runtime}
The PipeMax runtime consists of two components to support model execution, as follows.
\vspace{-3pt}
\subsubsection{Scheduler-coordinated Model Executor}
\vspace{-3pt}
The PipeMax runtime includes a scheduler-coordinated model executor that receives execution metadata from the centralized engine and executes scheduled model computation,
asynchronously transferring intermediate activations across pipeline stages via peer-to-peer communication.

\vspace{-2pt}
\subsubsection{Transfer-Efficient KV cache Engine}
\vspace{-2pt}
PipeMax designs a transfer-efficient KV cache engine to efficiently utilize available CPU--GPU bandwidth. 
Similar to existing systems such as vLLM~\cite{kwon2023efficient} and sglang~\cite{zheng2024sglang}, PipeMax stores KV cache on both GPU and CPU memory using the PagedAttention format, and further enhances KV cache transfers as follows:

\begin{figure}[!t]
    \hspace{0.5em}
    \vspace{-8pt}
    \includegraphics[width=\linewidth]{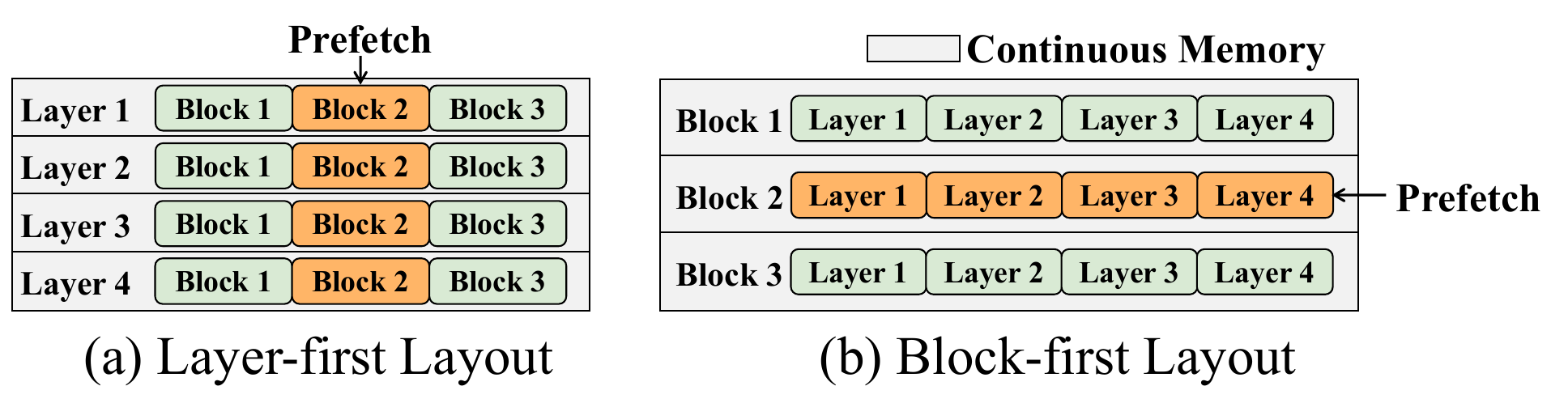}
    \vspace{-8pt}
    \caption{Layer-first vs block-first KV cache layouts for prefetching. The block-first layout enables contiguous block-level prefetching, while the layer-first layout incurs fragmented transfers.}
    \label{fig:layout-design}
    \vspace{-13pt}
\end{figure}

\vspace{-8pt}
\paragraph{Block-First Layout for Prefetching}
\textbf{C3} in Section~\ref{opportunities} identifies that existing PagedAttention-based designs adopt layer-first KV cache layouts and further partition them into blocks, causing KV cache prefetching to be sliced along both the layer and block dimensions.
As illustrated in Fig.~\ref{fig:layout-design}(a), when prefetching a block of a target request, the corresponding KV cache is not stored contiguously in memory, leading to inefficient CPU--GPU data transfers.

PipeMax instead adopts a block-first KV cache layout to accelerate prefetching.
Fig.~\ref{fig:layout-design}(b) shows that this layout colocates the KV cache of all layers within the same block into contiguous memory regions.
Thus, KV cache prefetching is sliced only along the block dimension, enabling efficient transfers.
\zhb{Notably, this layout change only affects KV cache storage and transfer, and remains fully compatible with PagedAttention without any implementation changes.}

\vspace{-8pt}
\paragraph{Asynchronous CPU-Assisted KV Cache Offloading}
\label{async-offload}
\begin{figure}[!b]
    \vspace{-10pt}
    \centering    \includegraphics[width=0.9\linewidth]{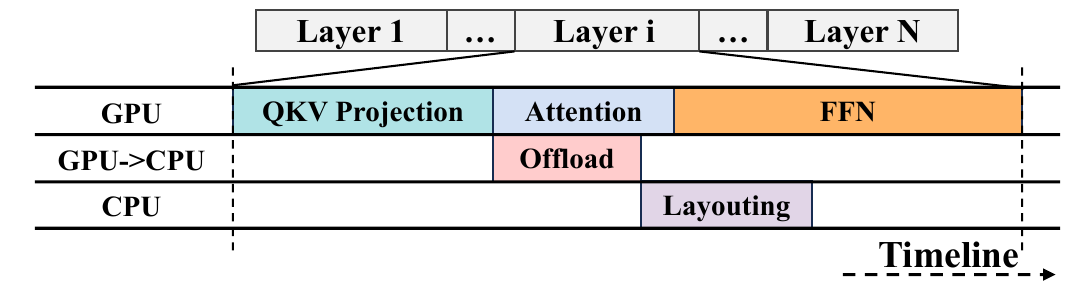}
    \vspace{-8pt}
    \caption{Execution timeline of per-layer KV cache generation, offloading, and CPU-side layouting.}
    \vspace{-8pt}
    \label{fig:async-kv-replication}
\end{figure}
Both  prefill and decode stages continuously generate KV cache that must be preserved in CPU-resident replicas, and ideally transferred in a way that can be overlapped with computation.
While the block-first layout improves prefetch efficiency, KV cache is generated on a per-layer basis, making block-wise offloading infeasible at generation time. 

Fortunately, right after QKV projection, each layer produces contiguous KV cache tensors before they are partitioned into PagedAttention blocks.
PipeMax exploits this observation by asynchronously offloading per-layer KV cache to CPU memory immediately after QKV computation, overlapping data transfer with subsequent computation.
In Fig.~\ref{fig:async-kv-replication}, the offloaded KV cache is then reorganized on the CPU side into a block-first layout for later prefetching, incurring negligible overhead.
Since PipeMax adopts a consistent block-first KV cache layout in both CPU and GPU memory, prefetched KV cache blocks can be directly transferred from CPU to GPU without any further reorganization.

\begin{figure}[!t]
    \vspace{-4pt}
    \centering    \includegraphics[width=0.86\linewidth]{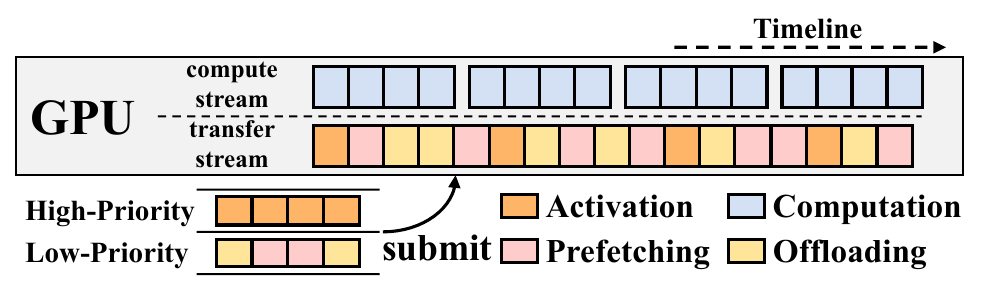}
    \vspace{-6pt}
    \caption{Priority-based transfer orchestration in PipeMax.}
    \label{fig:priority}
    \vspace{-12pt}
\end{figure}

\vspace{-11pt}
\paragraph{Demand-Priority Transfer Orchestration}

The PipeMax runtime involves three types of data movement: inter-stage activation transfers, KV cache prefetching, and KV cache offloading.
On commodity GPU servers, all transfers are carried out over PCIe, a full-duplex CPU--GPU interconnect, where prefetching and offloading proceed in opposite directions without interfering with each other.
Activation transfers occupy both PCIe directions and lie on the critical path of the next pipeline stage, whereas prefetching and offloading use only the CPU-to-GPU and GPU-to-CPU directions, respectively.
Although activation transfers involve much smaller data volumes and incur negligible impact on prefetching and offloading, they are latency-critical for pipeline execution and highly sensitive to interference from concurrent prefetching and offloading.

Inspired by priority-aware load management in AptMoE~\cite{wei2024aptmoe}, PipeMax employs priority-based transfer orchestration with multiple priority queues, prioritizing activation transfers without degrading KV cache throughput.
Fig.~\ref{fig:priority} shows that PipeMax maintains priority queues for issuing data transfer requests.
PipeMax enforces priorities by controlling the submission order of PCIe transfer requests: activation transfers are issued promptly, while KV cache prefetching and offloading are scheduled opportunistically.

\vspace{-12pt}
\section{Evaluation}
\vspace{-2pt}
We implement PipeMax based on vLLM v0.7.3.
To demonstrate its effectiveness, we evaluate PipeMax across diverse hardware configurations and workloads, and compare it with state-of-the-art approaches.
In addition, we conduct ablation studies to quantify the impact of individual components.
\vspace{-5pt}
\subsection{Experimental Setup}
\subsubsection{Node Testbed.}
\vspace{-2pt}

We conduct experiments on three nodes: two commodity GPU servers with 8×RTX~5090 and 8×L20 GPUs (both without NVLink), and one data-center server with 8×H100 GPUs interconnected via NVLink\footnote{H100 is included solely as a high-end reference platform, while our primary focus is on commodity GPU servers.}.
All GPUs connect to independent PCIe root complexes: RTX~5090 and H100 use PCIe~5.0 (64~GB/s), while L20 uses PCIe~4.0 (32~GB/s).

\vspace{-3pt}
\subsubsection{Model and Dataset Setup.}
\label{setup:workload}
\vspace{-3pt}
Table~\ref{tab:testbed-models} summarizes the deployed dense and MoE models, each scaled to the limits of its GPU platform.

\begin{table}[!t]
\centering
\caption{GPU platforms and deployed models.}
\vspace{-7pt}
\resizebox{0.83\linewidth}{!}{%
\begin{tabular}{cc}
\toprule
\textbf{Node} & \textbf{Workload (Models)} \\
\midrule
8 $\times$ RTX~5090 & LLaMA~2~70B (\textbf{70B}), Mixtral-8$\times$7B(\textbf{8$\times$7B})\\
8 $\times$ L20     &  LLaMA~2~70B (\textbf{70B}), Mixtral-8$\times$7B(\textbf{8$\times$7B})\\
8 $\times$ H100    & Qwen3~235B-A22B (\textbf{235B-A22B}) \\
\bottomrule
\end{tabular}
}
\label{tab:testbed-models}
\vspace{-15pt}
\end{table}

We select two representative datasets with distinct sequence-length characteristics.
\textbf{ShareGPT} features balanced input and output lengths, while
\textbf{LongBench} targets long-form inputs.
Table~\ref{tab:dataset-stats} reports their input and output length statistics.

\vspace{-3pt}
\subsubsection{Baseline setup.}
\vspace{-3pt}
\begin{figure}[!b]
    \vspace{-16pt}
    \centering    \includegraphics[width=0.985\linewidth]{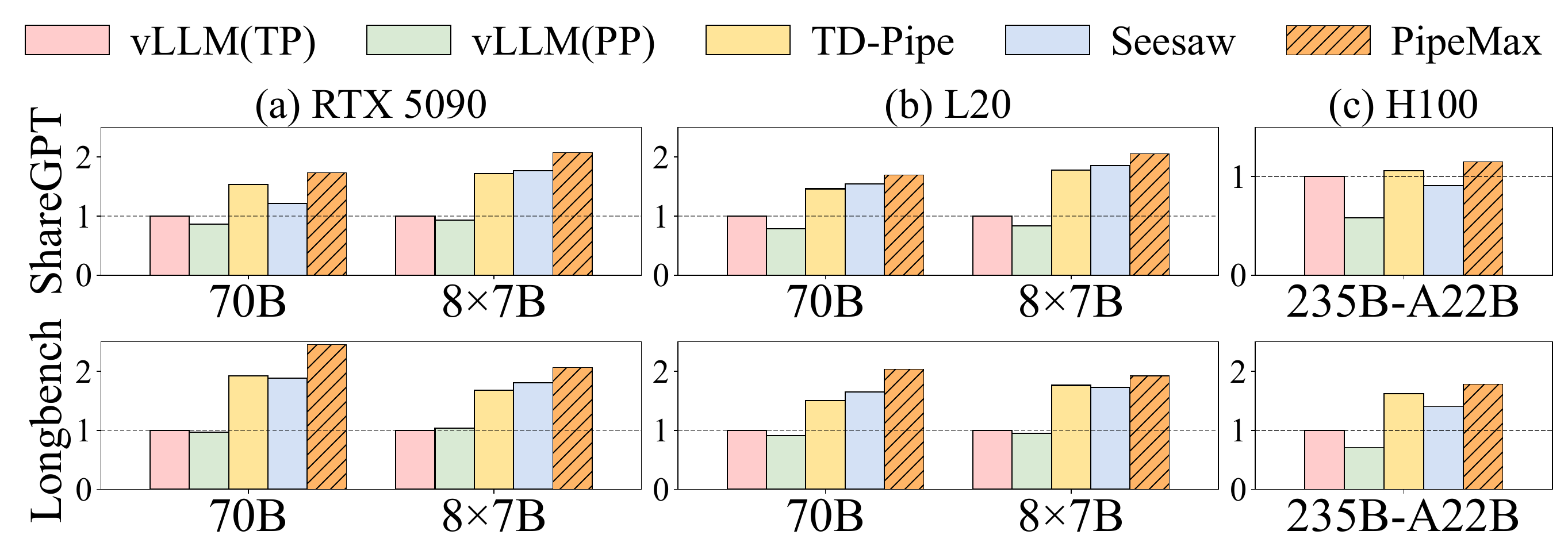}
    \vspace{-4pt}
    \caption{Normalized overall throughput (tokens/s) across workloads and GPU servers, where vLLM (TP) is normalized to 1.}
    \label{fig:overall}
\end{figure}
We compare PipeMax with the following baselines. Offloading-based systems such as FlexGen and TightLLM primarily target single-GPU inference and are not designed for multi-GPU execution, and thus are not included.

\vspace{-1pt}
\textbf{vLLM}~\cite{kwon2023efficient} is a widely adopted LLM inference engine supporting multiple parallelism strategies.

\vspace{-2pt}
\textbf{TD-Pipe}~\cite{zhang2025td} extends vLLM by addressing fundamental inefficiencies in pipeline parallelism, temporally decoupling prefill and decode to mitigate prefill–decode imbalance (Fig.~\ref{fig:prefill-decode}) and employing inter-batch work stealing to alleviate inter-batch imbalance (Fig.~\ref{fig:decode-pp}).

\vspace{-2pt}
\textbf{Seesaw}~\cite{su2025seesaw} builds upon vLLM by adopting pipeline parallelism during the prefill stage and prioritizing tensor parallelism during the decode stage.

We evaluate vLLM under both tensor parallelism(TP) and  pipeline parallelism(PP). 
TD-Pipe uses PP, while Seesaw follows its design, using PP for prefill and TP for decode.

All baselines are implemented on top of vLLM v0.7.3 to ensure a fair comparison.
Since the designs of Seesaw, TD-Pipe, and PipeMax are orthogonal to vLLM’s ongoing evolution, the relative performance trends reported in this paper are expected to remain valid on newer vLLM versions.

\vspace{-3pt}
\subsection{Overall Throughput}
\vspace{-3pt}

We compare PipeMax with all baselines in throughput, measured in tokens per second, across diverse hardware configurations and workloads (Section~\ref{setup:workload}).
Fig.~\ref{fig:overall} reports the normalized overall throughput results.
PipeMax outperforms vLLM(TP), vLLM(PP), TD-Pipe, and Seesaw by up to 2.45×, 2.51×, 1.42×, and 1.38×, respectively.

PipeMax leverages pipeline parallelism to reduce inter-GPU communication overhead compared to vLLM(TP), while eliminating the decode-phase anti-locality in prior pipeline-based designs such as TD-Pipe.
Although vLLM (PP) also adopts pipeline parallelism, its decode phase suffers from pronounced inter-batch imbalance(Fig.~\ref{fig:decode-pp}) due to lack of load balancing, which limits overall performance.
Seesaw shows lower-than-expected performance in our setting, as its decode phase relies on all-reduce communication, constrained by bandwidth in large-scale GPU configurations.
Moreover, even on NVLink-equipped H100 servers, the extreme compute capability of H100 GPUs makes inter-GPU communication under tensor parallelism non-trivial.
Consequently, PipeMax has the potential to deliver benefits on data-center GPU servers as well.

\vspace{-6pt}
\subsection{Ablation Study}
\vspace{-5pt}
\begin{table}[!t]
\centering
\caption{Length statistics of the evaluation datasets.}
\vspace{-7pt}
\resizebox{0.9\linewidth}{!}{%
\begin{tabular}{ccccc}
\toprule
\textbf{Dataset} 
& \textbf{$In_{\text{Avg}}$} 
& \textbf{$In_{\text{Med}}$} 
& \textbf{$Out_{\text{Avg}}$} 
& \textbf{$Out_{\text{Med}}$} \\
\midrule
ShareGPT~\cite{sharegpt2025} 
& 343.76 & 148.00 & 237.20 & 152.00 \\
LongBench~\cite{bai2024longbench} 
& 2686.89 & 2736.50 & 101.78 & 19.00 \\
\bottomrule
\end{tabular}
}
\label{tab:dataset-stats}
\vspace{-18pt}
\end{table}
This section studies the impact of design strategies in PipeMax via ablation experiments on RTX~5090 and L20 GPU servers with the 70B model and ShareGPT dataset.
\vspace{-5pt}
\subsubsection{Centralized Engine}
\vspace{-3pt}
\paragraph{Decode Execution-Time Estimator}
\begin{figure}[!b]
    \vspace{-20pt}
    \centering    \includegraphics[width=0.95\linewidth]{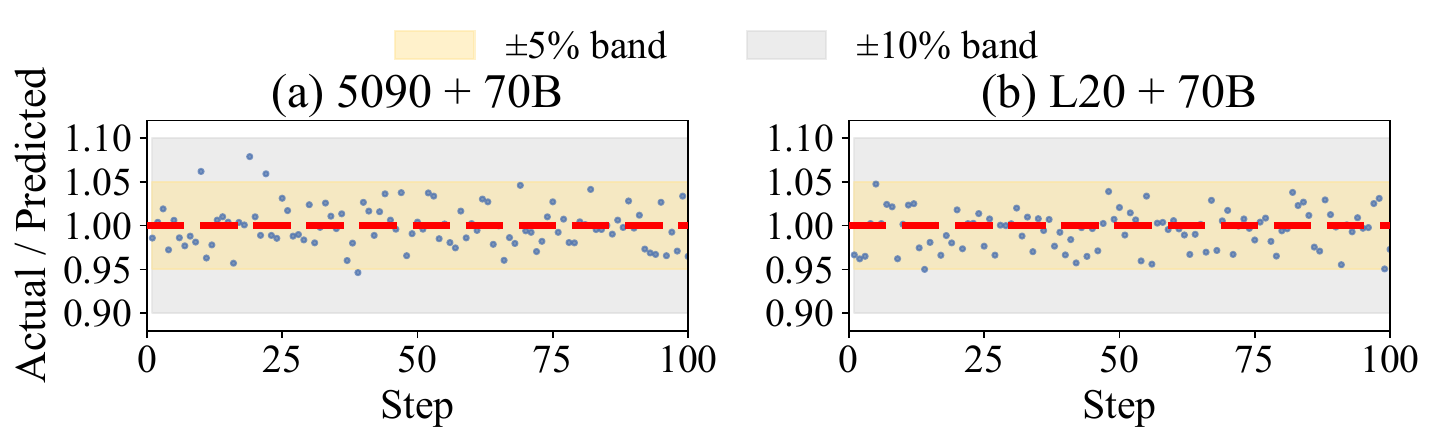}
    \vspace{-8pt}
    \caption{Ratio of actual to predicted decode execution time over 100 consecutive steps.}
    \label{fig:predict}
\end{figure}

To evaluate the accuracy of the decode execution-time estimator, we sample 100 consecutive decode steps during runtime and measure the ratio between the actual batch execution time and the predicted execution time at each step.
Figure \ref{fig:predict} reports these ratios for two representative workloads.
Across both workloads, predictions closely match actual execution times, with over 90\% of samples within 5\% error and worst-case deviation below 8\%.
This confirms the estimator’s accuracy and suitability for prefetch-aware scheduling.

\vspace{-10pt}
\paragraph{Prefetch-aware Decode Scheduler}

To evaluate the effectiveness of the prefetch-aware decode scheduler, we compare it against static prefetching baselines.
These baselines are implemented by replacing the original decode scheduler with static prefetching policies that prefetch a fixed fraction of available GPU memory, ranging from 5\% to 25\%, while keeping all other system components unchanged.
\begin{figure}[!t]
    \centering    \includegraphics[width=0.95\linewidth]{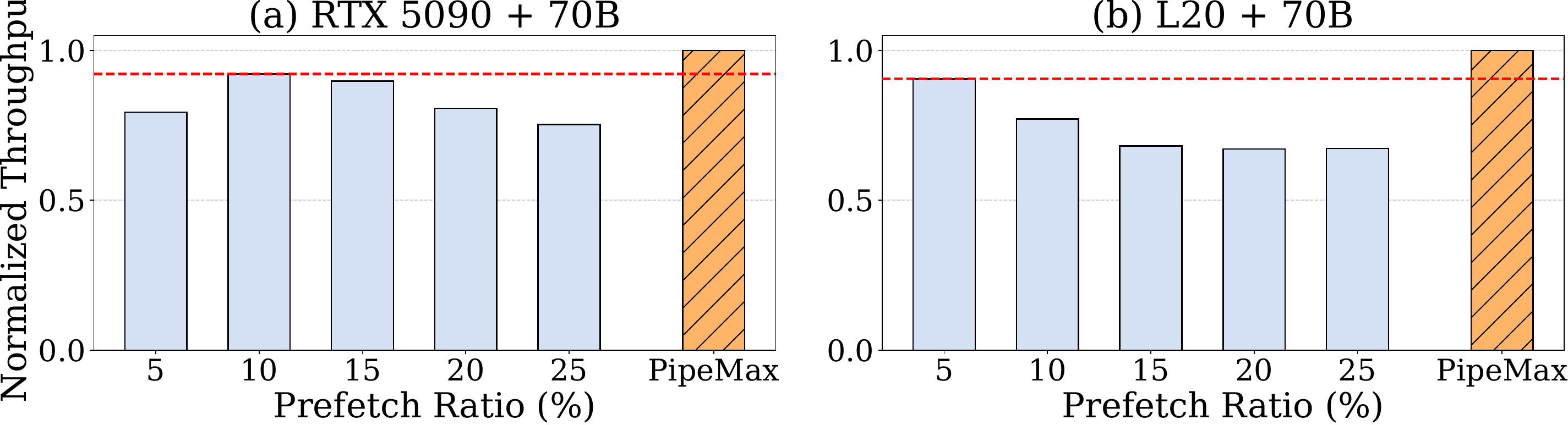}
    \vspace{-5pt}
    \caption{Normalized throughput of PipeMax vs. static prefetching with fixed prefetch ratios.}
    \label{fig:static}
    \vspace{-15pt}
\end{figure}

As shown in Fig.~\ref{fig:static}, PipeMax consistently outperforms the static prefetching baselines.
This result indicates that PipeMax can dynamically overlap KV cache prefetching with model execution, whereas static prefetching fails to fully utilize available overlap opportunities due to mismatches between prefetching decisions and actual execution progress, resulting in either insufficient prefetching or overly long prefetch operations that interrupt model execution.
\begin{figure}[!b]
    \vspace{-15pt}
    \centering    \includegraphics[width=0.9\linewidth]{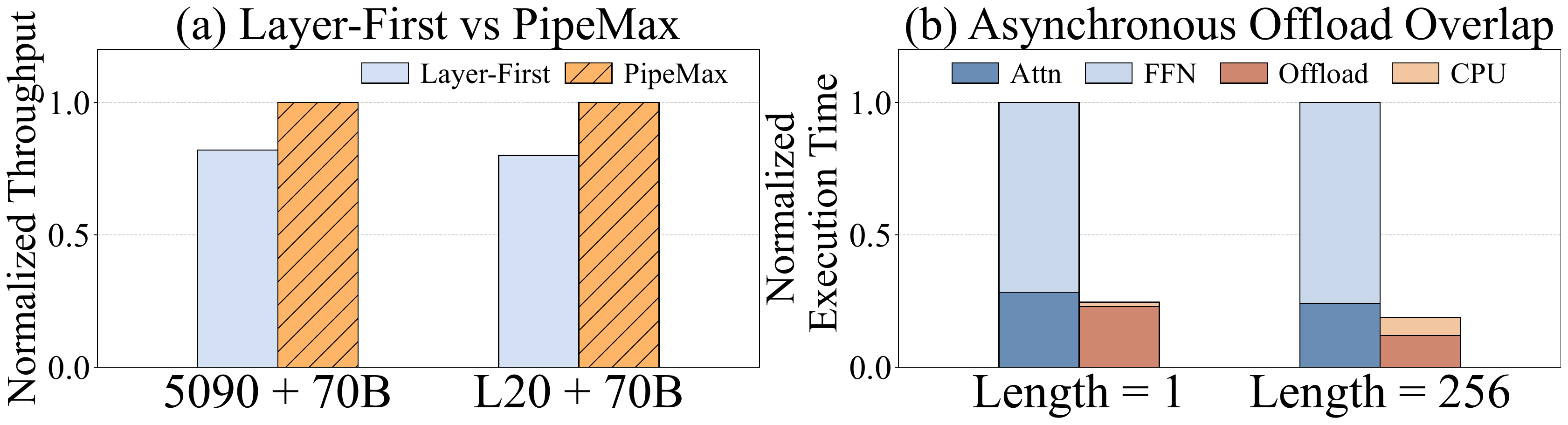}
    \vspace{-5pt}
    \caption{Runtime ablation of PipeMax: KV cache layout and offloading.}

    \label{fig:runtime}
    \vspace{-18pt}
\end{figure}

\vspace{-5pt}
\subsubsection{PipeMax Runtime}
\vspace{-3pt}

\paragraph{Block-First Layout}
We replace PipeMax’s block-first KV cache layout with a layer-first layout to evaluate the effectiveness of the block-first design.
As shown in Fig.~\ref{fig:runtime}(a), the block-first layout consistently outperforms the layer-first layout, as it enables contiguous memory allocation that improves prefetch efficiency and allows larger KV cache prefetching, leading to greater effective memory capacity.

We further measure the PCIe bandwidth utilization of both layouts.
The results show that under the block-first layout, with the block size set to the vLLM default of 16, KV cache prefetching can saturate nearly 90\% of the available PCIe bandwidth.
In contrast, the layer-first layout achieves only about 30\% bandwidth utilization.

\vspace{-10pt}
\paragraph{Asynchronous Offloading}

To validate that asynchronous offloading in Section~\ref{async-offload} can be hidden by computation, we measure the execution-time breakdown of computation and KV cache offloading.

Fig.~\ref{fig:runtime}(b) reports two representative prefill cases with input lengths of 1 and 256 tokens on RTX~5090 with the 70B model. In both cases, KV cache offloading and CPU-side processing are fully overlapped with attention and FFN computation; intermediate input lengths show similar behavior but are omitted for brevity.
During decode, different batch sizes (reflecting decode lengths) exhibit the same trend, with larger attention cost further masking offloading latency.
\vspace{-4pt}
\subsection{Runtime Dynamics during Decode}
\vspace{-4pt}
\begin{figure}[!t]
    \centering
    \includegraphics[width=0.94\linewidth]{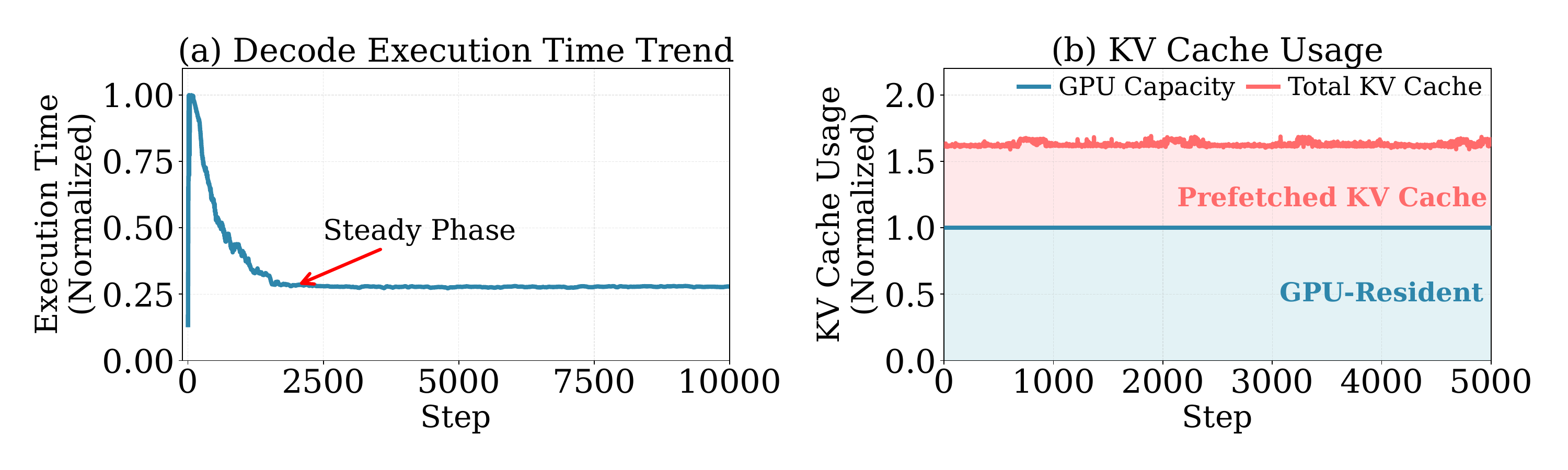}
    \vspace{-12pt}
    \caption{Decode runtime dynamics of PipeMax.}
    \label{fig:time_and_memory}
    \vspace{-17pt}
\end{figure}

In this section, we show the runtime behavior of PipeMax on RTX~5090 with the 70B model using the ShareGPT dataset.

Fig.~\ref{fig:time_and_memory}(a) shows that decode execution time increases rapidly after the prefill-to-decode transition due to short-request prefetching, then converges as GPU memory and PCIe bandwidth become limiting factors (Section~\ref{scheduler}).
In the steady phase, PipeMax balances batch workloads to prevent inter-batch imbalance.
Fig.~\ref{fig:time_and_memory}(b) shows the KV cache footprint per batch in this phase.
With total GPU memory normalized to 1.0, prefetched KV cache occupies a substantial fraction of GPU memory, indicating PipeMax effectively expands usable KV cache capacity via prefetching.

\vspace{-10pt}
\section{Related Work}
\textbf{LLM Inference.}
\vspace{-1pt}
LLM inference has attracted growing attention, motivating extensive system-level optimizations.
Orca~\cite{280922} introduces continuous batching, while vLLM~\cite{kwon2023efficient} proposes PagedAttention for efficient KV cache management.
Subsequent work optimizes LLM inference for both online and offline scenarios.
For online serving, systems such as DistServe~\cite{zhong2024distserve}, Sarathi-Serve~\cite{agrawal2023sarathi}, EcoServe~\cite{du2025ecoserve}, and Bullet~\cite{lin2025bullet} mitigate prefill--decode interference to improve service quality.
For offline inference, beyond the parallelism and offloading techniques studied here, BatchLLM~\cite{zheng2024batchllm} and BlendServe~\cite{zhao2024blendserve} improve throughput via prefix sharing.
Prefix sharing is orthogonal to PipeMax and can be seamlessly combined with our approach.

\textbf{Pipeline Parallelism Enhanced by Offloading.} 
Pipeline parallelism and offloading both improve resource utilization via concurrent batch execution, making their integration natural.
Prior work explored this combination for training, including Mobius~\cite{feng2023mobius}, APT-MoE~\cite{wei2024aptmoe}, and PipeOffload~\cite{wan2025pipeoffload}, which overlap pipeline execution with data transfers via stage, expert, or activation offloading.

\vspace{-10pt}
\section{Conclusion}
\vspace{-4pt}
This paper presents PipeMax, a high-throughput LLM inference system for commodity GPU servers.
PipeMax boosts pipeline-parallel inference by offloading inactive KV cache and dynamically scheduling computation and KV cache movement to maximize compute–data overlap.
Experiments show that PipeMax outperforms state-of-the-art LLM systems by up to 2.51×, 1.42×, and 1.38× on 8 GPUs.

\section*{Impact Statement}
This paper presents work whose goal is to advance the field of Large Language Model Inference Systems. There are many potential societal consequences of our work, none of which we feel must be specifically highlighted here.
\balance
\bibliography{example_paper}
\bibliographystyle{paper}

\clearpage
\makeatletter
\let\balance\relax
\let\endbalance\relax
\let\flushend\relax
\let\endflushend\relax
\makeatother

\onecolumn
\appendix
\newpage
\appendix
\section{Appendix}
\subsection{Pipeline Parallelism for the Prefill Stage}
\label{app:prefill-proof}

In this section, we present a proof for the total execution-time formula of the pure prefill pipeline shown in Fig.~\ref{fig:pure-prefill} (also Stated in Section~\ref{pp-analyse-prefill}).

\paragraph{Setup.}
Consider $m$ independent prefill requests executed on an $n$-stage pipeline (i.e., $n$ workers/GPUs).
Let $t_i$ denote the per-stage execution time of request $i$ during prefill, 
\emph{assuming a homogeneous pipeline where all stages have identical execution time for a given request}.
\footnote{Prefill has no cross-request data dependency; thus, the first stage can be kept continuously busy by scheduling ready requests.}
Let $T(m,n)$ denote the total makespan to process all $m$ requests on the $n$-stage pipeline.

\paragraph{Assumption (No bubble at Stage 1).}
We assume the first stage is fully utilized (i.e., no idle bubble at Stage~1). This can be Ensured by continuously dispatching ready requests whose dependencies (if any) have been resolved.

\begin{theorem}[Prefill pipeline execution time]
\label{thm:prefill-time}
Under the above assumption, the total execution time satisfies
\begin{equation}
T(m,n)
=
\sum_{i=1}^{m} t_i
+
(n-1)\cdot \max_{1\le i\le m} t_i.
\end{equation}
\end{theorem}

\begin{proof}
Define
\begin{equation}
f(m,n) \triangleq \sum_{i=1}^{m} t_i + (n-1)\cdot \max_{1\le i\le m} t_i.
\end{equation}
Our goal is to show that the total execution time satisfies
\begin{equation}
T(m,n) = f(m,n),
\end{equation}
for all $m>0$ and $n>0$.

We prove the theorem by induction.

\paragraph{Base cases.}
When $n=1$, the pipeline degenerates to sequential execution:
\begin{equation}
T(m,1)=\sum_{i=1}^{m} t_i = f(m,1).
\end{equation}

When $m=1$, the single request traverses all $n$ stages:
\begin{equation}
T(1,n)=n\cdot t_1 = f(1,n).
\end{equation}

\paragraph{Inductive step.}
Given the base cases, we assume that the equality
\begin{equation}
T(i,j)=f(i,j), \qquad \forall\, 1 \le i \le m,\; 1 \le j \le n.
\label{eq:inductive-hypothesis}
\end{equation}

It suffices to show that the equality is preserved when extending the
pipeline by one stage or one request, i.e.,
\begin{equation}
T(m,n+1)=f(m,n+1)
\quad\text{and}\quad
T(m+1,n)=f(m+1,n).
\end{equation}

\textbf{We first prove that $T(m,n+1)=f(m,n+1)$.}

Fix the number of pipeline stages to $(n+1)$ and consider increasing the
number of requests.
The equality holds for a single request, since
\begin{equation}
T(1,n+1) = (n+1)\cdot t_1 = f(1,n+1).
\end{equation}
It therefore suffices to show that, for any $k\in[1,m-1]$,
\begin{equation}
T(k,n+1)=f(k,n+1)
\;\Longrightarrow\;
T(k+1,n+1)=f(k+1,n+1).
\label{eq: case1}
\end{equation}
Once this implication is established, $T(m,n+1)=f(m,n+1)$ follows directly.

Consider the pipeline dependency for processing the $(k+1)$-th request on an $(n+1)$-stage pipeline.
Let $T(k,n+1)$ denote the time when the $k$-th request finishes stage $(n+1)$, and $T(k+1,n)$ denote the time when request $(k+1)$ finishes stage $n$.
Due to pipeline precedence constraints, the execution of request $(k\!+\!1)$ at stage $(n\!+\!1)$ is governed by two dependencies, as illustrated in Fig.~\ref{fig:prefill-dependencies}(a) and (b).

\begin{figure}[H]
    \centering
    \begin{minipage}{0.48\linewidth}
        \centering
        \includegraphics[width=\linewidth]{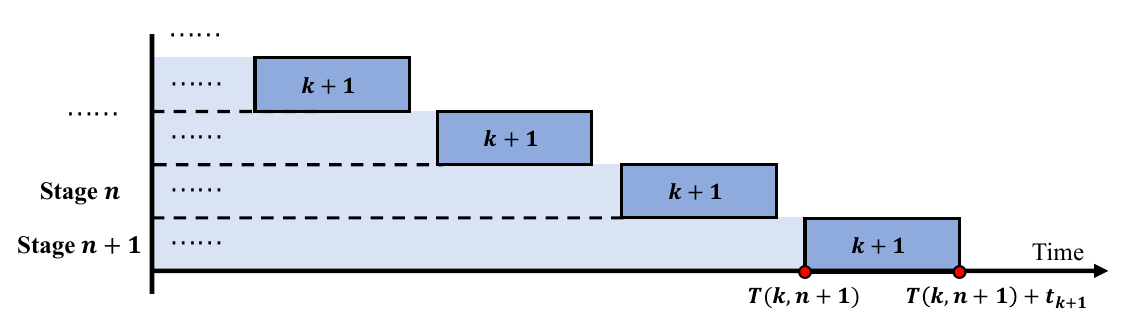}
        \vspace{2pt}
        \textbf{(a) Case 1: request-order dependency.}
    \end{minipage}
    \hfill
    \begin{minipage}{0.48\linewidth}
        \centering
        \includegraphics[width=\linewidth]{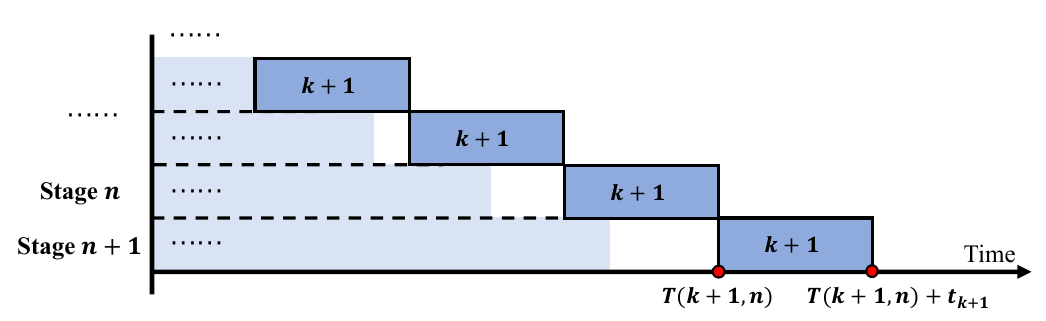}
        \vspace{2pt}
        \textbf{(b) Case 2: stage-order dependency.}
    \end{minipage}
    \caption{Precedence constraints when extending the pipeline by one
    additional request, which determine the start time of request $(k+1)$ at
    stage $(n+1)$.}
    \label{fig:prefill-dependencies}
\end{figure}

Combining the above two cases, the completion time of request $(k\!+\!1)$ at stage $(n\!+\!1)$ is determined by the later of the two dependencies, and thus satisfies
\begin{equation}
T(k\!+\!1,n\!+\!1)
= \max\!\bigl\{T(k,n\!+\!1),\, T(k\!+\!1,n)\bigr\} + t_{k+1}.
\label{eq:prefill-recurrence}
\end{equation}

Define $M_k \triangleq \max_{1\le i\le k} t_i$.

\textbf{Case 1: $t_{k+1}\le M_k$.}

In this case, $M_{k+1}=M_k$.
Using the induction hypothesis,
\begin{align}
T(k,n+1) &= \sum_{i=1}^{k} t_i + n\cdot M_k, \\
T(k+1,n)
&= \sum_{i=1}^{k+1} t_i + (n-1)\cdot M_{k+1} = \sum_{i=1}^{k+1} t_i + (n-1)\cdot M_k.
\end{align}
Thus,
\begin{equation}
\begin{aligned}
T(k,n+1) - T(k+1,n)
&= nM_k - \bigl(t_{k+1} + (n-1)M_k\bigr) \\
&= M_k - t_{k+1} \\
&\ge 0.
\end{aligned}
\end{equation}

so $\max\{T(k,n+1),T(k+1,n)\}=T(k,n+1)$.

Plugging into \eqref{eq:prefill-recurrence},
\begin{align}
T(k+1,n+1)
&= T(k,n+1) + t_{k+1} \nonumber\\
&= \sum_{i=1}^{k+1} t_i + n\cdot M_{k+1} \nonumber\\
&= f(k+1,n+1).
\end{align}

\textbf{Case 2: $t_{k+1} > M_k$.}

In this case, $M_{k+1}=t_{k+1}$.
Using the induction hypothesis,
\begin{align}
T(k,n+1) &= \sum_{i=1}^{k} t_i + n\cdot M_k, \\
T(k+1,n)
&= \sum_{i=1}^{k+1} t_i + (n-1)\cdot M_{k+1}
 = \sum_{i=1}^{k+1} t_i + (n-1)\cdot t_{k+1}.
\end{align}
Thus,
\begin{equation}
\begin{aligned}
T(k+1,n) - T(k,n+1)
&= \Bigl(\sum_{i=1}^{k+1} t_i + (n-1)t_{k+1}\Bigr)
 - \Bigl(\sum_{i=1}^{k} t_i + nM_k\Bigr) \\
&= nt_{k+1} - nM_k \\
&= n(t_{k+1}-M_k) \\
&> 0.
\end{aligned}
\end{equation}
Therefore, $\max\{T(k,n+1),T(k+1,n)\}=T(k+1,n)$.
Plugging into \eqref{eq:prefill-recurrence},
\begin{align}
T(k+1,n+1)
&= T(k+1,n) + t_{k+1} \nonumber\\
&= \sum_{i=1}^{k+1} t_i + n\cdot t_{k+1} \nonumber\\
&= \sum_{i=1}^{k+1} t_i + n\cdot M_{k+1} \nonumber\\
&= f(k+1,n+1).
\end{align}

In both cases, the implication in \eqref{eq: case1} holds, and therefore
$T(m,n+1)=f(m,n+1)$ follows.


\textbf{We next prove that $T(m+1,n)=f(m+1,n)$.}

Similar to the previous case, we fix the number of requests to $(m+1)$ and consider increasing the number of pipeline stages.

The equality holds for a single stage, since
\begin{equation}
T(m+1,1)=\sum_{i=1}^{m+1} t_i = f(m+1,1).
\end{equation}
It therefore suffices to show that, for any $k\in[1,n-1]$,
\begin{equation}
T(m+1,k)=f(m+1,k)
\;\Longrightarrow\;
T(m+1,k+1)=f(m+1,k+1).
\label{eq:case2}
\end{equation}

Consider processing request $(m+1)$ on a $(k+1)$-stage pipeline.
Let $T(m,k+1)$ denote the time when request $m$ finishes stage $(k+1)$,
and let $T(m+1,k)$ denote the time when request $(m+1)$ finishes stage $k$.
By pipeline precedence constraints (illustrated in Fig.~\ref{fig:prefill-dependencies-v2}),
request $(m+1)$ can start stage $(k+1)$ only after both dependencies are resolved.
Therefore, its completion time satisfies
\begin{equation}
T(m+1,k+1)=\max\!\bigl\{T(m,k+1),\,T(m+1,k)\bigr\}+t_{m+1}.
\label{eq:prefill-recurrence-v2}
\end{equation}


\begin{figure}[H]
    \centering
    \begin{minipage}{0.48\linewidth}
        \centering
        \includegraphics[width=\linewidth]{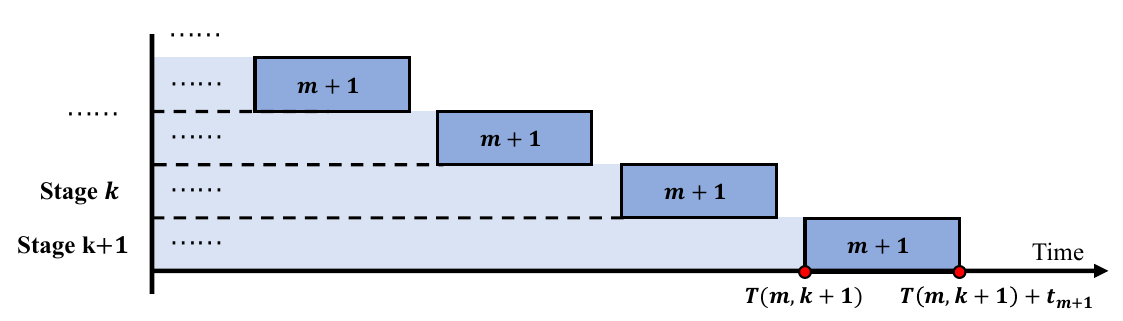}
        \vspace{2pt}
        \textbf{(a) Case 1: request-order dependency.}
    \end{minipage}
    \hfill
    \begin{minipage}{0.48\linewidth}
        \centering
        \includegraphics[width=\linewidth]{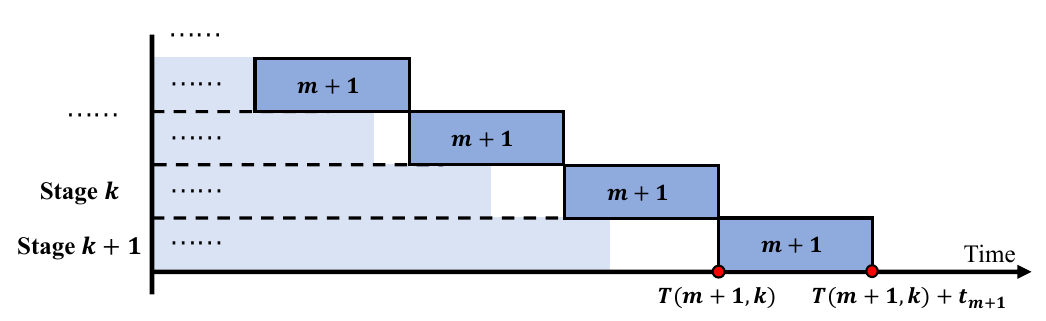}
        \vspace{2pt}
        \textbf{(b) Case 2: stage-order dependency.}
    \end{minipage}
    \caption{Precedence constraints for request $(m+1)$ when extending the pipeline by one additional stage, which determine its start time at stage $(k+1)$.}

    \label{fig:prefill-dependencies-v2}
\end{figure}

The dependency constraints in \eqref{eq:prefill-recurrence-v2} are structurally
identical to those in \eqref{eq:prefill-recurrence}, with the roles of requests
and pipeline stages exchanged.
Accordingly, the same case analysis applies here.
By distinguishing whether $t_{m+1}$ is no larger than or exceeds
$\max_{1\le i\le m} t_i$, we can show that

\begin{equation}
T(m+1,k+1)=f(m+1,k+1).
\end{equation}
This establishes the implication in \eqref{eq:case2}, and therefore
$T(m+1,n)=f(m+1,n)$ holds.

\textbf{Conclusion}

Since the equality holds for the base cases and is preserved when extending either the number of requests or the number of pipeline stages, we conclude that $T(m,n)=f(m,n)$ holds for all $m>0$ and $n>0$.

\end{proof}

\paragraph{Implication.}
The theorem implies that, once the first stage is kept bubble-free, the overall prefill makespan is dominated by
(i) the cumulative work injected into Stage~1 and
(ii) a fixed drain cost of $(n-1)\max_i t_i$.
When $m\gg n$, the drain term becomes amortized, and the total time is effectively governed by Stage~1 throughput.
Therefore, optimizing prefill reduces to keeping Stage~1 continuously saturated, which naturally connects to our batch construction and KV-cache memory management design.

\subsection{Pipeline Parallelism for the Decode Stage}
\label{app:decode-proof}

In this section, we derive the average token budget per decode batch under pipeline parallelism.

Let $M$ denote the per-GPU memory capacity, $W$ the total model weight size,
$n$ the pipeline degree, and $T$ the KV Cache size per token.

Since the model weights are evenly partitioned across pipeline stages, the
per-GPU weight footprint is
\begin{equation}
W_{\mathrm{gpu}} = \frac{W}{n}.
\end{equation}
Accordingly, the memory available for KV Cache on each GPU is
\begin{equation}
M_{\mathrm{kv}} = M - W_{\mathrm{gpu}}.
\end{equation}

Since the KV Cache of each token is sharded across all $n$ GPUs, the
per-token KV Cache footprint on each GPU is
\begin{equation}
T_{\mathrm{gpu}} = \frac{T}{n}.
\end{equation}

Maintaining full pipeline utilization during decode requires $n$ concurrent
batches to remain resident in GPU memory.
Let $N_{\text{token}}^{\max}$ denote the system-wide maximum number of
storable tokens, which is given by
\begin{equation}
\begin{aligned}
N_{\text{token}}^{\max}
&= \frac{M_{\mathrm{kv}}}{T_{\mathrm{gpu}}} \\
&= \frac{M - W/n}{T/n} \\
&= \frac{nM - W}{T}.
\end{aligned}
\end{equation}

Dividing this system-wide token capacity evenly across the $n$ resident decode
batches yields the average token budget per batch:
\begin{equation}
\begin{aligned}
N_{\text{avg}} 
&= \frac{N_{\text{token}}^{\max}}{n}\\
&= \frac{M - W/n}{T}.
\end{aligned}
\end{equation}

\subsection{Scheduler Algorithm}
\label{app:scheduler}

Here we present the detailed scheduling procedure and the greedy algorithm for selecting $P_t$.

\subsubsection{Scheduling Procedure}

The scheduling procedure described in Section~\ref{scheduler} is summarized in Algorithm~\ref{alg:decode-scheduler}.
\begin{algorithm}[H]
\caption{Prefetch-Aware Decode Scheduler}
\label{alg:decode-scheduler}
\begin{algorithmic}[1]
\Require Decode request set $\mathcal{R}$, pipeline depth $n$
\Ensure Iterative decode batches $\mathcal{D} = \{D_0,\dots,D_{n-1}\}$

\State \textbf{Initialization:}
\State Partition $\mathcal{R}$ into $n$ initial decode batches $\mathcal{D}$
\State Initialize iteration counter $t \gets 0$
\State Initialize steady-phase flag $\mathsf{steady} \gets \textbf{false}$

\State \textbf{Iterative Scheduling:}
\While{decode not finished}
    \State $i_t \gets t \bmod n$
    \State $j_t \gets (i_t + 1) \bmod n$
    \State $k_t \gets (i_t - 1) \bmod n$
    
    \State Predict the execution time $\hat{T}_t$ of batch $D_{i_t}$
    \State Derive the prefetch budget $\mathcal{B}_t \gets B \cdot \hat{T}_t$
    
    \State Update steady-phase flag $\mathsf{steady}$ by monitoring whether $\mathcal{B}_t$ has stabilized
    
    \State Retain GPU-resident requests in the next batch:
    \State \hspace{1em}$D_{j_t}^{\mathrm{res}} \gets D_{j_t} \cap \mathcal{R}_t$
    
    \If{\textbf{not} $\mathsf{steady}$}
        \State \textit{// Ramp-up phase: prioritize short requests to fill the budget}
        \State \hspace{1em}$P_t \gets \textsc{ShortFirstFill}(\mathcal{R}_t^{\mathrm{cpu}}, \mathcal{B}_t)$
    \Else
        \State \textit{// Steady phase: match execution-time gap}
        \State \hspace{1em}$P_t \gets \textsc{GreedySelect}(\mathcal{R}_t^{\mathrm{cpu}}, \mathcal{B}_t, \Delta \hat{T}_t)$
    \EndIf
    
    \State Prefetch the KV cache of requests in $P_t$ from CPU to GPU
    \State Reclaim GPU memory by overwriting KV cache blocks of the inactive batch $D_{k_t}$
    
    \State Update the next decode batch:
    \State \hspace{1em}$D_{j_t} \gets D_{j_t}^{\mathrm{res}} \cup P_t$
    
    \State Submit batch $D_{i_t}$ for execution and batch $D_{j_t}$ for prefetching to the PipeMax runtime
    \State $t \gets t + 1$
\EndWhile
\end{algorithmic}
\end{algorithm}

The scheduler operates in an iterative manner and adapts its behavior across iterations based on the evolution of the prefetch budget $\mathcal{B}_t$.

At each iteration, PipeMax first predicts the execution time of the currently executing batch and derives the corresponding prefetch budget.
Requests whose KV cache already resides in GPU memory are retained in the next batch, while additional CPU-resident requests are selected for prefetching to fully utilize the available budget.

PipeMax distinguishes the warm-up phase from the steady phase by tracking the stabilization of the prefetch budget.
During the warm-up phase, the scheduler prioritizes short requests with smaller prefix lengths to best-effort utilize the limited prefetch budget and rapidly increase decode execution time.
As longer requests are gradually admitted, execution time and the prefetch budget may exhibit temporary fluctuations due to GPU memory and CPU--GPU bandwidth constraints.
Eventually, the system converges to a bounded execution regime, after which PipeMax enters the steady phase.

Once the system reaches a steady phase, the scheduler switches to a prefetch-aware selection policy that matches the execution-time contribution of prefetched requests to the remaining time gap, thereby balancing decode batches and avoiding inter-batch imbalance.

In practice, PipeMax detects the onset of the steady phase by tracking the evolution of the prefetch budget $\mathcal{B}_t$ over a sliding window of recent iterations.
At the beginning of decoding, $\mathcal{B}_t$ typically grows rapidly as execution time ramps up.
As the system transitions into the steady phase, $\mathcal{B}_t$ stabilizes and fluctuates within a narrow range.
Specifically, PipeMax maintains a sliding window of the most recent $w$ iterations and considers the system to have entered the steady phase when the relative variation of $\mathcal{B}_t$ within the window falls below a predefined threshold.

\subsubsection{Greedy Algorithm for Selecting $P_t$}

After PipeMax reaches a steady phase, it strives to maintain stable execution times across iterations to reduce inter-batch imbalance.

\paragraph{Problem Formulation.}
At iteration $t$, PipeMax selects a set of CPU-resident requests $P_t$ to augment the next decode batch $D_{j_t}$.
Each request $r$ is associated with a prefix length $L_r$ and contributes an execution-time cost $\alpha + \beta L_r$ according to Eq.~(\ref{model}).
Given the prefetch budget $\mathcal{B}_t = B \cdot \hat{T}_t$ and the remaining execution-time gap $\Delta \hat{T}_t$ defined in Eq.~(\ref{cpu-time}), the goal is to select a subset $P_t$ such that
\[
\sum_{r \in P_t} L_r \approx \mathcal{B}_t,
\]
i.e., the selected requests aim to nearly saturate the prefetch budget without exceeding it,
while making the total execution-time contribution
\[
\sum_{r \in P_t} (\alpha + \beta L_r)
\]
as close as possible to $\Delta \hat{T}_t$.

\paragraph{Greedy Algorithm}

To select the set of CPU-resident requests $P_t$, PipeMax adopts a two-stage heuristic that combines greedy selection with exchange-based local refinement.

The greedy stage prioritizes length utilization by selecting requests in a length-first manner, rapidly saturating the prefetch budget to obtain a near-feasible initial solution.
Under the execution-time model, this effectively drives the length-dependent term $\beta \sum_{r \in P_t} L_r$ toward its budget-limited maximum.

Building on this initial solution, the refinement stage performs limited exchange operations that replace a small number of selected requests with unselected ones, adjusting the batch cardinality while preserving the budget constraint.
Under the length budget, the cumulative length term $\beta \sum_{r \in P_t} L_r$ remains largely fixed after greedy initialization.
As a result, the refinement primarily adjusts the constant per-request component $\alpha |P_t|$ of the execution-time model.
These local exchanges are guided by the execution-time model and aim to minimize the mismatch between the modeled batch execution time
\[
\alpha |P_t| + \beta \sum_{r \in P_t} L_r
\]
and the remaining execution-time gap $\Delta \hat{T}_t$.

Concretely, when the modeled execution time falls short of the target, the refinement replaces longer requests with multiple shorter ones of comparable total length, increasing the $\alpha |P_t|$ term.
Conversely, when the modeled execution time exceeds the target, multiple shorter requests are replaced by a longer one to reduce the $\alpha |P_t|$ contribution, while keeping the total prefix length approximately unchanged.
This exchange process continues until the mismatch falls below a predefined threshold or a maximum number of refinement steps is reached.

\begin{algorithm}[H]
\caption{Greedy Selection with Local Refinement}
\label{alg:greedy-refine}
\begin{algorithmic}[1]
\Require CPU-resident requests $\mathcal{R}_t^{\mathrm{cpu}}$ with prefix length $L_r$
\Require Prefetch budget $\mathcal{B}_t$, target gap $\Delta \hat{T}_t$
\Ensure Selected subset $P_t$

\State \textbf{Greedy Initialization:}
\State Sort $\mathcal{R}_t^{\mathrm{cpu}}$ in descending order of $L_r$
\State $P_t \gets \emptyset$, $S \gets 0$ \Comment{$S$: cumulative prefix length}
\For{each request $r \in \mathcal{R}_t^{\mathrm{cpu}}$}
    \If{$S + L_r \le \mathcal{B}_t$}
        \State $P_t \gets P_t \cup \{r\}$
        \State $S \gets S + L_r$
    \EndIf
\EndFor

\State \textbf{Local Refinement:}
\For{a fixed number of refinement steps}
    \State $T \gets \alpha |P_t| + \beta S$
    \If{$|T - \Delta \hat{T}_t|$ is below a threshold}
        \State \textbf{break}
    \EndIf
    \If{$T < \Delta \hat{T}_t$}
        \State Replace one long selected request with multiple shorter unselected ones
        \State to increase $|P_t|$ while keeping $S \approx \mathcal{B}_t$
    \ElsIf{$T > \Delta \hat{T}_t$}
        \State Replace multiple short selected requests with one longer unselected one
        \State to decrease $|P_t|$ while keeping $S \approx \mathcal{B}_t$
    \EndIf

\EndFor

\State \Return $P_t$
\end{algorithmic}
\end{algorithm}

\end{document}